\documentclass[11pt]{article}
\usepackage[T1]{fontenc}
\usepackage{lmodern}
\usepackage[a4paper,margin=28mm]{geometry}
\usepackage{amsmath,amssymb,amsthm}
\usepackage{booktabs}
\usepackage{graphicx}
\usepackage[colorlinks=true,linkcolor=blue,citecolor=blue,urlcolor=blue]{hyperref}
\usepackage{microtype}

\makeatletter
\renewcommand{\@seccntformat}[1]{\csname the#1\endcsname.\quad}
\makeatother

\newtheorem{theorem}{Theorem}[section]
\newtheorem{lemma}[theorem]{Lemma}
\newtheorem{proposition}[theorem]{Proposition}
\newtheorem{corollary}[theorem]{Corollary}
\newcommand{\CC}{\mathbb C}
\newcommand{\doi}[1]{\href{https://doi.org/#1}{\nolinkurl{#1}}}
\newcommand{\rank}{\operatorname{rank}}
\newcommand{\degen}{\unrhd}
\newcommand{\cmm}[3]{\langle #1,#2,#3\rangle}

\title{Rectangular matrix multiplication from polynomial products}
\author{Przemysław Uznański\\[2pt]\small Pathway\\[2pt]
\small\href{mailto:przemek@pathway.com}{\texttt{przemek@pathway.com}}}
\date{October 9, 2026}

\begin{document}
\maketitle

\begin{abstract}
In this note, we extend the matrix-multiplication analysis in~\cite{openai-2026-matrix-nine-fourths} to
rectangular products and prove
\[
 \omega(1,k,1)\le
 \begin{cases}
 2,&0\le k\le \frac{1}{2},\\
 1+k+\frac{1}{4k},&k\ge \frac{1}{2}.
 \end{cases}
\]
In particular, $\omega(1,\frac{1}{2},1)=2$ and the dual exponent satisfies
$\alpha\ge \frac{1}{2}$. The proof uses explicit degenerations of convolution tensors,
multinomial counting, and Strassen's spectral theorem.
An angular constraint on three dot-product parameters gives both the
rectangular curve and a bound for $\omega(a,b,c)$, with exact exponent
$a+b$ whenever $a\ge b\ge c\ge0$ and $c\le\frac{ab}{a+b}$.
As an application, Zwick's algorithm
\cite{zwick-2002-bridging-sets} for all-pairs shortest paths in directed
unweighted graphs runs in $O(n^{2.5})$ time.
Combining the rectangular bound with the $(\min,+)$-product improvement
of Alman and Vassilevska Williams~\cite{alman-williams-2026} further gives
$O(n^{2.4999})$ running time.
\end{abstract}

\section{Introduction}\label{sec:statement}

\noindent\textbf{Disclaimer.} This note, its proofs, analysis and writeup
were obtained fully through OpenAI Codex (GPT-6.1). The author has read and analyzed this note in full
and takes full responsibility for its contents, and takes no credit for this work.

The initial motivation for this note was to explore whether the methods
of the groundbreaking matrix-multiplication result
in~\cite{openai-2026-matrix-nine-fourths} could be pushed to analyze
rectangular matrix products. We also aim to gather useful bounds and
their algorithmic consequences in one place. Rephrasing the proofs in
a form accessible to the theoretical computer science community has
value in its own right.

\paragraph{Complexity convention.}
Throughout this note, all complexity bounds suppress a $+o(1)$ term
in the exponent: $O(n^c)$ denotes $n^{c+o(1)}$. All shape parameters
are fixed.

Rectangular matrix multiplication is a basic ingredient in algorithms
for all-pairs problems, including shortest paths, lowest common ancestors,
and Hamming distance. Its exponents depend on the shape of the product;
an improvement for square matrices need not give a comparable
improvement for every rectangular shape.
The analysis in~\cite{openai-2026-matrix-nine-fourths} establishes the square bound
$\omega\le2.25$ using a shared-leg entropy inequality and degenerations
of convolution tensors. We extend that analysis to
rectangular products by retaining the orientation information lost
under full leg symmetrization.

We work over $\CC$. The resulting exponent bounds hold over every
field of characteristic zero by invariance of asymptotic tensor rank
under field extension~\cite[Theorem~3.10]{strassen-1988-asymptotic-spectrum}.

For fixed nonnegative real $u,v,w$, let $\omega(u,v,w)$ be the infimum
of the exponents $e$ for which multiplying an
$\lceil n^u\rceil\times\lceil n^v\rceil$ matrix by a
$\lceil n^v\rceil\times\lceil n^w\rceil$ matrix uses
$O(n^e)$ arithmetic operations, with the convention above.
Write $\omega(k)=\omega(1,k,1)$ and
$\alpha=\sup\{k\ge0:\omega(k)=2\}$.

We first state the bound for general matrix shapes, with the exponents
sorted in decreasing order. We interpret the threshold $\frac{ab}{a+b}$
as zero when $a=b=0$.

\begin{theorem}\label{cor:general-exponents}
For every fixed $a\ge b\ge c\ge0$,
\begin{equation}\label{eq:general-exponents}
 \omega(a,b,c)\le
 \begin{cases}
 a+b,&c\le\dfrac{ab}{a+b},\\[4pt]
 \dfrac{a+b+c}{2}
 +\dfrac{1}{4}\left(\dfrac{ab}{c}+\dfrac{ac}{b}+\dfrac{bc}{a}\right),
 &c>\dfrac{ab}{a+b}.
 \end{cases}
\end{equation}
When $c\le\frac{ab}{a+b}$, equality holds by the size lower bound.
\end{theorem}

Applying this bound to the sorted dimensions of $(1,k,1)$ gives
the rectangular curve.

\begin{corollary}\label{thm:curve}
For every fixed real $k\ge0$,
\begin{equation}\label{eq:curve}
 \omega(k)\le
 \begin{cases}
 2,&0\le k\le \frac{1}{2},\\[2pt]
 1+k+\dfrac1{4k},&k\ge\frac{1}{2}.
 \end{cases}
\end{equation}
Consequently $\omega(k)=2$ for $0\le k\le\frac{1}{2}$.
\end{corollary}

Corollary~\ref{thm:curve} gives exponent two for inner dimension
up to $n^{0.5}$, and an upper bound within $\frac{1}{4k}$ of the input-size
lower bound $1+k$ for $k\ge1$.

\begin{figure}[tbp]
 \centering
 \includegraphics[width=0.86\linewidth]{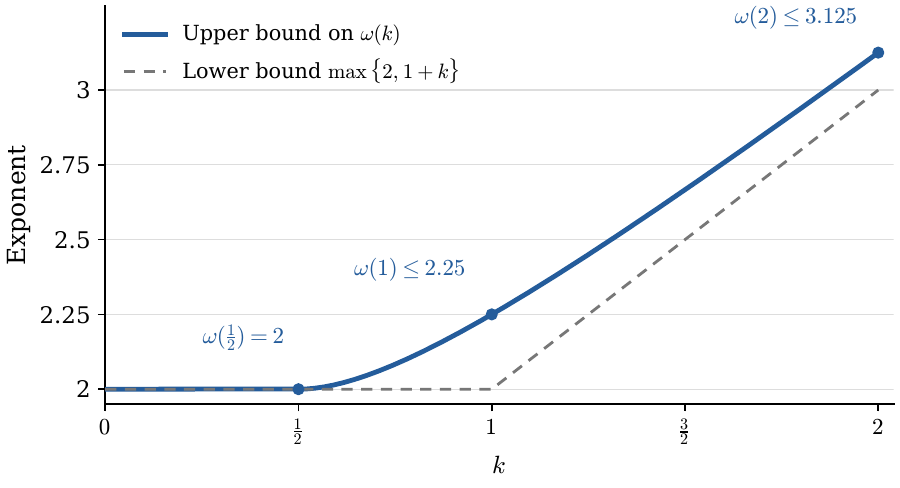}
 \caption{The upper bound on $\omega(k)$ from Corollary~\ref{thm:curve} for
 $0\le k\le2$, together with the size lower bound
 $\max\{2,1+k\}$. The bounds coincide for
 $0\le k\le\frac{1}{2}$.}
 \label{fig:rectangular-curve}
\end{figure}

The proof has two parts. First, we list explicit ways to split and
combine polynomial products and apply a tensor character to them.
A character assigns numerical sizes that add on direct sums,
multiply on tensor products, and cannot increase under restriction.
Multinomial counting gives concavity and growth inequalities for these
sizes. Rescaling, common slopes, and affine bounds
then give $\sum_i\arccos\sqrt{p_i}\ge\pi/2$ for the character's three
dot-product parameters.

Second, Wolstenholme's inequality bounds weighted sums of these parameters.
We apply Strassen's spectral
theorem: the largest character value equals asymptotic tensor rank.
We state precisely the form of this theorem we use, then show how the
rank bound gives a running-time bound for every fixed matrix shape.

\paragraph{Relation to the preceding analysis.}
The two polynomial splittings in Lemma~\ref{lem:openai-inputs} recap
the proofs of Lemmas~4.1 and~4.2 of~\cite{openai-2026-matrix-nine-fourths}
in coefficient coordinates. The finite separation adapts Proposition~3.1
of that reference. All three constructions are proved in the appendices.
Multinomial counting then gives the required power-sum inequality.
The replication construction, rescaling, and coupling of six
affine offsets supply the rectangular extension.
The weighted bound uses the classical Wolstenholme inequality
\cite{liu-2020-weighted-trigonometric}.
Strassen's spectral theorem is a standard input; a modern exposition is
\cite{wigderson-zuiddam-2025-asymptotic-spectra}. We cite it rather than
reproduce its proof. We make no priority claim for the standard tools or
the finite construction variants. Section~\ref{sec:applications}
substitutes the bounds into known algorithmic tradeoffs.

\section{Polynomial products and finite constructions}\label{sec:constructions}

We begin by describing matrix multiplication and convolution as tensors.
We then build ways to split polynomial products and separate computations
that share variables. These constructions give the inequalities in
Section~\ref{sec:discrete}. A further construction replicates a convolution;
it will let us compare convolutions at different scales in
Section~\ref{sec:homogeneous}.

A tensor is a trilinear form in three groups of variables, called its
\emph{legs}. A \emph{restriction} makes linear substitutions independently
on the three legs. A \emph{degeneration} is a limit of such restrictions:
the substitutions may depend on $\varepsilon\ne0$, and every resulting
coefficient converges as $\varepsilon\to0$. Our constructions first change
bases and then multiply variables by integer powers of $\varepsilon$.
The resulting coefficient array is a polynomial in $\varepsilon$ whose
constant coefficient is the desired tensor.
We write $T\ge S$ for a restriction and $T\degen S$ for a degeneration.
Successive constructions can use a single parameter: replace the earlier
parameter by a sufficiently large power of the later one, so every
earlier error term still has positive total exponent.

The $\rank(T)$ is the least number of products of three linear
forms, one per leg, summing to $T$. The direct sum $S\oplus T$ uses
disjoint variables on all three legs and represents independent computations.
The tensor product $S\otimes T$ pairs their indices, as in a Kronecker
product, and represents their composition on blocks. We write $mT$ for
$m$ independent copies, and $T^{\otimes N}$ for $N$ tensor factors.
Thus $\rank(T)\,xyz\ge T$. All factors used below restrict to $xyz$
by keeping one nonzero term; this is what it means to drop a factor.

The oriented dot products are
\[
 B_X(q)=x\sum_{j=1}^q y_jz_j,
\]
with $B_Y(q),B_Z(q)$ obtained by permuting the legs. They have rank $q$,
and $B_i(q)\otimes B_i(q')$ is $B_i(qq')$ after reindexing. The matrix
multiplication tensor factors as
\begin{equation}\label{eq:mm-factorization}
 \cmm{A}{B}{C}
 =\sum_{u,v,w}x_{uv}y_{vw}z_{wu}
 \cong B_Y(A)\otimes B_Z(B)\otimes B_X(C).
\end{equation}
Thus the three dot-product dimensions specify the matrix shape.

For distinct input legs $i,j$, with remaining output leg $k$, the
\emph{convolution tensor} is
\[
 C_{ij}(r,s)=\sum_{h=0}^{r-1}\sum_{\ell=0}^{s-1}
 x_h^{(i)}x_\ell^{(j)}x_{h+\ell}^{(k)}.
\]
It multiplies two polynomials represented by their coefficient vectors.
Its rank is $r+s-1$: evaluation at $r+s-1$ points and interpolation
give the upper bound, and the independent output coefficients give the
matching matrix-rank lower bound. Decreasing either input length is a
restriction. Also
\begin{equation}\label{eq:boundary-convolutions}
 C_{ij}(1,s)\cong B_i(s),\qquad C_{ij}(r,1)\cong B_j(r).
\end{equation}

Blocks $T_h$ are \emph{matched sectors sharing leg $i$} when
$T=\sum_h T_h$ and the blocks may share variables on leg $i$, but use
pairwise disjoint variables on each other leg. The following constructions
hold for every permutation of the legs.

\begin{lemma}[Finite separation and polynomial splittings]\label{lem:openai-inputs}
\leavevmode
\begin{enumerate}
\item If $T$ has $M$ matched sectors sharing leg $i$, then
\begin{equation}\label{eq:finite-separation}
 (2M)T\degen\bigoplus_{h=1}^M\bigl(T_h\otimes B_i(M)\bigr).
\end{equation}
The sectors on the right are independent on all three legs.
\item For $r\ge1$ and $s\ge2$, $C_{ij}(r,s)\otimes B_i(2)$ degenerates
to two matched sectors sharing leg $i$, namely
$C_{ij}(r,s+1)$ and $C_{ij}(r,s-1)$.
\item For $r,s\ge1$, $C_{ij}(r,3s+r-1)$ degenerates to three matched
sectors sharing leg $i$: two copies of $C_{ij}(r,s)$ and one copy of
$C_{ik}(r,s)$.
\end{enumerate}
\end{lemma}

Appendices~\ref{app:determinant-filtration} and~\ref{app:tripling} give
the polynomial splittings in coefficient coordinates.
Appendix~\ref{app:finite-separation} proves
\eqref{eq:finite-separation} using a finite geometric sum and a nonnegative
square. The next construction splits a convolution into repeated
independent copies of a smaller convolution.

\begin{lemma}[Replicating a convolution]\label{lem:replication}
For all positive integers $m,r,s$,
\begin{equation}\label{eq:replication}
 C_{ij}(mr,ms)\degen mC_{ij}(r,s).
\end{equation}
\end{lemma}

The construction combines polynomial division with the standard evaluation
description of cyclic convolution
\cite[Theorem~5 and Section~5.1]{ju-solomonik-2020-convolution}.
We give the full proof in Appendix~\ref{app:replication}.

\section{Tensor characters and counting}\label{sec:discrete}

To use the preceding constructions, we assign numerical values to tensors
using tensor characters. The splittings and separation give concavity
and tripling inequalities for their convolution values, which we rescale
in the next section. We also identify the three parameters that determine
their matrix-multiplication values and state Strassen's theorem: bounding
every character value will give an asymptotic rank bound in
Section~\ref{sec:exponents}.

A \emph{tensor character} is a function $\lambda$ from tensors to
nonnegative real numbers satisfying
\begin{equation}\label{eq:character-rules}
\begin{aligned}
 \lambda(S\oplus T)&=\lambda(S)+\lambda(T),&
 \lambda(S\otimes T)&=\lambda(S)\lambda(T),\\
 \lambda(xyz)&=1,&
 S\ge T&\ \Longrightarrow\ \lambda(S)\ge\lambda(T).
\end{aligned}
\end{equation}
These rules express addition of independent computations, multiplication
of tensor factors, and monotonicity under linear substitutions. Matrix
rank after grouping one leg against the other two is an example.
Keeping one nonzero term and using $\rank(T)\,xyz\ge T$ gives
$1\le\lambda(T)\le\rank(T)$ for every nonzero $T$.

\begin{lemma}[Interpolation of a degeneration]\label{lem:degeneration-interpolation}
Suppose a degeneration $T\degen S$ produces a tensor $P(\varepsilon)$
whose coefficients are polynomials of degree at most $L$. Then, for
every integer $N\ge1$,
\begin{equation}\label{eq:degeneration-interpolation}
 (NL+1)T^{\otimes N}\ge S^{\otimes N}.
\end{equation}
Consequently, every tensor character satisfies $\lambda(S)\le\lambda(T)$.
\end{lemma}

Equation~\eqref{eq:degeneration-interpolation} is the standard interpolation
conversion from degeneration to restriction
\cite[Theorem~8]{christandl-jensen-zuiddam-2018-nonmultiplicativity}, applied
to tensor powers. We include the proof for completeness in
Appendix~\ref{app:degeneration-interpolation}.

The \emph{asymptotic rank} of $T$ is
\[
 \widetilde R(T)=\lim_{N\to\infty}\rank(T^{\otimes N})^{1/N}
 =\inf_{N\ge1}\rank(T^{\otimes N})^{1/N}.
\]
The limit exists because tensor rank is submultiplicative; apply the
subadditive-sequence form of Fekete's lemma to its logarithm.

\begin{theorem}[Strassen's spectral theorem]\label{thm:spectral}
For every nonzero tensor $T$ over $\CC$,
\begin{equation}\label{eq:spectral}
 \widetilde R(T)=\max_{\lambda}\lambda(T),
\end{equation}
where the maximum ranges over all tensor characters.
\end{theorem}

We use \cite[Theorem~3.8]{strassen-1988-asymptotic-spectrum}; see also
\cite[Corollary~3.30 of the author version]{wigderson-zuiddam-2025-asymptotic-spectra}.
The easy direction follows from
$\lambda(T)^N=\lambda(T^{\otimes N})\le\rank(T^{\otimes N})$.
The theorem supplies the opposite direction. Thus, to prove an upper
bound on asymptotic rank, it suffices to bound every character value.

Fix an arbitrary character $\lambda$. Its \emph{dot-product parameters}
are $p_i=\log_2\lambda(B_i(2))\in[0,1]$, for $i=X,Y,Z$. They determine
the values at all integer dimensions:
\begin{equation}\label{eq:dot-parameters}
 \lambda(B_i(q))=q^{p_i}\qquad(q\ge1).
\end{equation}
The interval $[0,1]$ follows from $1\le\lambda(B_i(2))\le2$.
Tensor powers turn multiplication of dimensions into multiplication of
character values; comparison with powers of two then gives
\eqref{eq:dot-parameters}. Appendix~\ref{app:counting-rescaling} gives
the details. In particular,
\eqref{eq:mm-factorization} gives
\begin{equation}\label{eq:mm-character}
 \lambda(\cmm{A}{B}{C})=A^{p_Y}B^{p_Z}C^{p_X}.
\end{equation}
Our aim is to constrain these three parameters for every character.
Define its \emph{convolution profiles} by
\begin{equation}\label{eq:discrete-profiles}
 f_{ij}(r,s)=\lambda(C_{ij}(r,s)).
\end{equation}

The next counting argument turns the polynomial splittings into
inequalities for these profiles.

Suppose $T$ splits into $s$ matched sectors $T_1,\ldots,T_s$ sharing
leg $i$. In $T^{\otimes N}$ a sector is specified by a word of $N$
sector labels. Retain on the two nonshared legs the words with $n_h$
occurrences of label $h$, where $\sum_h n_h=N$. There are
\[
 M=\binom{N}{n_1,\ldots,n_s}
\]
matched sectors. For example, multiplicities $(2,1)$ give the words
$112,121,211$. Each sector is a product of the $T_h$ with the prescribed
multiplicities. Apply \eqref{eq:finite-separation}; now all legs
distinguish the words, so their factors can be reordered independently.
This gives
\begin{equation}\label{eq:type-shared}
 (2M)T^{\otimes N}\degen
 M\left(\bigotimes_h T_h^{\otimes n_h}\right)\otimes B_i(M).
\end{equation}

\begin{lemma}[Multinomial power-sum inequality]\label{lem:power-sum}
If $T\degen\sum_{h=1}^s T_h$ is a splitting into matched sectors
sharing leg $i$ and $p_i>0$, then
\begin{equation}\label{eq:power-sum}
 \lambda(T)^{1/p_i}\ge\sum_h\lambda(T_h)^{1/p_i}.
\end{equation}
For fully independent summands, $T\degen\bigoplus_h T_h$ gives
$\lambda(T)\ge\sum_h\lambda(T_h)$ directly by additivity.
\end{lemma}

The counting idea is to sum over all multiplicities. With
$a_h=\lambda(T_h)^{1/p_i}$, applying the character to
\eqref{eq:type-shared} and using the multinomial theorem gives
\[
 \left(\sum_h a_h\right)^N
 \le (N+1)^s2^{1/p_i}\lambda(T)^{N/p_i}.
\]
The number of multiplicity tuples is polynomial in $N$, so it disappears
when we take $N$-th roots and let $N\to\infty$.
The full calculation is in Appendix~\ref{app:counting-rescaling}.

Rank bounds and restrictions give
\begin{equation}\label{eq:profile-rank}
 1\le f_{ij}(r,s)\le r+s-1,
 \qquad f_{ij}(r,s)=f_{ji}(s,r),
\end{equation}
and coordinatewise monotonicity. The equality merely interchanges the
two polynomial inputs of the same convolution.

\begin{lemma}[Powered concavity and tripling]\label{lem:discrete}
For $p_i>0$, the sequence $s\mapsto f_{ij}(r,s)^{1/p_i}$ is concave,
and for all integers $r,s\ge1$,
\begin{equation}\label{eq:discrete-tripling}
 f_{ij}(r,3s+r-1)^{1/p_i}
 \ge2f_{ij}(r,s)^{1/p_i}+f_{ik}(r,s)^{1/p_i},
\end{equation}
where $i,j,k$ are distinct. For $p_j>0$, the sequence
$r\mapsto f_{ij}(r,s)^{1/p_j}$ is also concave.
\end{lemma}

For example, the two-sector splitting gives
\[
 2f_{ij}(r,s)^{1/p_i}
 \ge f_{ij}(r,s+1)^{1/p_i}+f_{ij}(r,s-1)^{1/p_i},
\]
which is discrete concavity. The three-sector splitting gives
\eqref{eq:discrete-tripling}; input interchange gives concavity in the
other coordinate. Appendix~\ref{app:counting-rescaling} supplies the proof.

\section{Rescaling the convolution profiles}\label{sec:homogeneous}

We next pass from integer input lengths to real ones, so we can compare
convolution shapes at any scale. Replication from
Section~\ref{sec:constructions} lets us enlarge both inputs together;
the bounds, monotonicity, concavity, and tripling inequalities from the
preceding section control the resulting limit. We obtain continuous
profiles that retain concavity and tripling and scale linearly when both
lengths grow by the same factor.

Applying a character to Lemma~\ref{lem:replication} gives
\begin{equation}\label{eq:profile-replication}
 f_{ij}(mr,ms)\ge m f_{ij}(r,s).
\end{equation}
This makes the rescaled values stable under further replication.
For real $r,s>0$, define
\begin{equation}\label{eq:rescaled-profile}
 g_{ij}(r,s)=\lim_{m\to\infty}
 \frac{f_{ij}(\lceil mr\rceil,\lceil ms\rceil)}{m},
 \qquad m\in\mathbb Z_{>0}.
\end{equation}
The following lemma justifies this limit and lets us work with real
lengths throughout the remaining analysis.

\begin{lemma}[Rescaled profiles]\label{lem:homogeneous}
The limit in \eqref{eq:rescaled-profile} exists. The profiles $g_{ij}$
are continuous, coordinatewise nondecreasing, and satisfy
\begin{equation}\label{eq:limit-rank}
 \min\{r,s\}\le g_{ij}(r,s)\le r+s,\qquad
 g_{ij}(r,s)=g_{ji}(s,r),\qquad
 g_{ij}(tr,ts)=t\,g_{ij}(r,s).
\end{equation}
For $p_i>0$, $g_{ij}^{1/p_i}$ is concave in $s$ and satisfies
\begin{equation}\label{eq:general-tripling}
 g_{ij}(r,3s+r)^{1/p_i}
 \ge2g_{ij}(r,s)^{1/p_i}+g_{ik}(r,s)^{1/p_i}.
\end{equation}
For $p_j>0$, $g_{ij}^{1/p_j}$ is concave in $r$.
All lengths and scale factors here may be arbitrary positive reals.
\end{lemma}

The discrete inequalities thus become inequalities about homogeneous
curves. The additive $-1$ in \eqref{eq:discrete-tripling} disappears at
this scale. The proof first compares rational lengths at a common integer
scale; homogeneity and monotonicity then give continuity and the real
extension. Appendix~\ref{app:counting-rescaling} includes that argument
and verifies that rounding the lengths in \eqref{eq:rescaled-profile}
does not change the limit.

\section{Common slopes and the angular constraint}\label{sec:slopes}

With the profile properties from Lemma~\ref{lem:homogeneous}, we can now
relate the three dot-product parameters $p_i\in[0,1]$. Comparing the growth
of differently oriented convolutions leads to the angular constraint in
Proposition~\ref{prop:angular-cap}. This is the only result from this
section needed for the running-time bounds in Section~\ref{sec:exponents}.

Index the three tensor legs by $1,2,3$.
Full proofs of this section's slope and growth
claims are in Appendix~\ref{app:slopes-angular}.

\begin{lemma}[Common slopes and affine bounds]\label{lem:intercepts}
Fix $i$ with $p_i>0$, and let $j,k$ be the other legs. The functions
$P_{ij}(s)=g_{ij}(1,s)^{1/p_i}$ and $P_{ik}(s)=g_{ik}(1,s)^{1/p_i}$
have a common asymptotic slope $L_i>0$. There are offsets satisfying
\begin{equation}\label{eq:intercept-budgets}
 c_{ij},c_{ik}\ge0,\qquad c_{ij}+c_{ik}\le1,
\end{equation}
such that, for all real $s,t>0$,
\[
 P_{ij}(s+t)\ge P_{ij}(s)+L_it,\qquad
 L_i s\le P_{ij}(s)\le L_i(s+c_{ij}),
\]
and likewise for $P_{ik}$.
\end{lemma}

Concavity gives each curve a limiting slope. The coupled tripling
inequalities force the two slopes to agree. Writing $R=P_{ij}+P_{ik}$,
they also give $R(3s+1)\ge3R(s)$; iteration yields
$R(s)\le L_i(2s+1)$. Consequently the two curves stay above their
common limiting line, and their normalized vertical offsets are
nondecreasing with sum at most one. Their limits are $c_{ij},c_{ik}$.
This explains both the affine bounds and the constraint on their sum.

\begin{proposition}[Angular constraint]\label{prop:angular-cap}
The dot-product parameters of every tensor character satisfy
\begin{equation}\label{eq:angular-cap}
 \sum_{i=1}^3\arccos\sqrt{p_i}\ge\frac{\pi}{2}.
\end{equation}
\end{proposition}

\begin{proof}[Proof outline]
If some $p_i=0$, its summand is already $\pi/2$. Assume all three
parameters are positive. For this outline, suppose $g=g_{ij}$ is
differentiable. The increment and upper bounds in
Lemma~\ref{lem:intercepts} give $P'_{ij}(t)\ge L_i$ and
$P_{ij}(t)\le L_i(t+c_{ij})$. Homogeneity gives
$g(r,s)=rP_{ij}(s/r)^{p_i}$, so, also interchanging the inputs,
\[
\begin{aligned}
 \frac{\partial_s g(r,s)}{g(r,s)}
 &=\frac{p_i}{r}\frac{P'_{ij}(s/r)}{P_{ij}(s/r)}
 \ge\frac{p_i}{s+c_{ij}r},\\
 \frac{\partial_r g(r,s)}{g(r,s)}
 &=\frac{p_j}{s}\frac{P'_{ji}(r/s)}{P_{ji}(r/s)}
 \ge\frac{p_j}{r+c_{ji}s}.
\end{aligned}
\]
Differentiating $g(tr,ts)=t\,g(r,s)$ at $t=1$ gives
\[
 r\frac{\partial_r g(r,s)}{g(r,s)}
 +s\frac{\partial_s g(r,s)}{g(r,s)}=1.
\]
Substituting the two lower bounds yields
\begin{equation}\label{eq:ratio-bound}
 1\ge\frac{p_i s}{s+c_{ij}r}+\frac{p_j r}{r+c_{ji}s}
 \qquad(r,s>0).
\end{equation}
Appendix~\ref{app:slopes-angular} obtains the same inequality using
finite increments, without assuming differentiability.
Clearing denominators gives a quadratic in $r/s$. Its nonnegativity
for every positive ratio implies
\begin{equation}\label{eq:intercept-products}
 \sqrt{c_{ij}c_{ji}}\ge
 \left(\sqrt{p_ip_j}-\sqrt{(1-p_i)(1-p_j)}\right)_+,
\end{equation}
where $(z)_+=\max\{z,0\}$.

For nonnegative $x_1,x_2,x_3$, AM--GM and
\eqref{eq:intercept-budgets} give
\begin{equation}\label{eq:angular-quadratic-form}
 \sum_{i<j}2\sqrt{c_{ij}c_{ji}}x_ix_j
 \le\sum_{i<j}(c_{ij}x_i^2+c_{ji}x_j^2)
 \le\sum_i x_i^2.
\end{equation}
Put $\theta_i=\arccos\sqrt{p_i}$ and suppose
$S=\theta_1+\theta_2+\theta_3<\pi/2$. Take $x_i=\cos\theta_i>0$.
Then \eqref{eq:intercept-products} gives
$\sqrt{c_{ij}c_{ji}}\ge\cos(\theta_i+\theta_j)$.
Using the product-to-sum identities, we get
\[
\begin{aligned}
 \cos(\theta_i+\theta_j)x_j+\cos(\theta_i+\theta_k)x_k
 &=\cos\theta_i+
   \frac{\cos(\theta_i+2\theta_j)+\cos(\theta_i+2\theta_k)}{2}\\
 &=x_i+\cos S\cos(\theta_j-\theta_k)>x_i.
\end{aligned}
\]
The extra term is positive: $\cos S>0$ and
$\cos(\theta_j-\theta_k)>0$, since $|\theta_j-\theta_k|\le S<\pi/2$.
Multiplying by $x_i$ and summing contradicts
\eqref{eq:angular-quadratic-form}. Appendix~\ref{app:slopes-angular}
contains the complete proof, including the quadratic calculation.
\end{proof}

\section{From character bounds to running times}\label{sec:exponents}\label{sec:general}

We can now turn the parameter constraint into matrix-multiplication
algorithms. The other ingredients are the matrix-tensor formula
\eqref{eq:mm-character} and Strassen's spectral theorem
(Theorem~\ref{thm:spectral}). For a given matrix shape, we bound a weighted
sum of the three parameters for every character; the theorem then bounds
asymptotic rank. We first show how this rank bound gives a running-time
bound, then apply the angular constraint from
Proposition~\ref{prop:angular-cap} to obtain the general-shape bound and
the rectangular curve.

\begin{lemma}[Character bounds give algorithms]\label{lem:cost-conversion}
Fix $\mathbf h=(h_X,h_Y,h_Z)\ge0$. If the dot-product parameters of
every tensor character satisfy $\sum_i h_ip_i\le E$, then
$\omega(h_Y,h_Z,h_X)\le E$.
\end{lemma}

\begin{proof}
The zero target is trivial. For a fixed integer $N\ge2$, put
$D_i=\lceil N^{h_i}\rceil$ and $M_N=\bigotimes_i B_i(D_i)$. Since
$D_i\le2N^{h_i}$ and $0\le p_i\le1$, every character satisfies
\[
 \lambda(M_N)=\prod_i D_i^{p_i}
 \le8N^{\sum_i h_ip_i}\le8N^E.
\]
Theorem~\ref{thm:spectral} therefore gives
\begin{equation}\label{eq:shape-asymptotic-rank}
 \widetilde R(M_N)\le8N^E.
\end{equation}

Fix any positive exponent slack $\eta$. First choose $N$ so large that
$8<N^{\eta/2}$. The infimum defining $\widetilde R(M_N)$ then gives
a fixed integer $K$ and an exact rank decomposition with
\[
 R:=\rank(M_N^{\otimes K})\le N^{K(E+\eta)}.
\]
This is a finite bilinear algorithm for matrix dimensions
$(D_Y^K,D_Z^K,D_X^K)$. Apply it recursively $\ell$ times to blocks.
By \eqref{eq:mm-factorization}, the resulting dimensions are
$(D_Y^{K\ell},D_Z^{K\ell},D_X^{K\ell})$ and there are $R^\ell$
scalar multiplications.

Each input and output dimension, viewed as a number of scalar entries,
is at most $R$, by the corresponding matrix-rank lower bound. Thus the
linear substitutions and additions over all recursion levels use
$O(\ell R^\ell)$ operations. Given input size $n$, choose
$\ell=\lceil\log_N n/K\rceil$ and pad to these dimensions. They cover
the target shape because $D_i^{K\ell}\ge n^{h_i}$, while
$R^\ell=O(n^{E+\eta})$. This proves the bound for every positive slack,
and hence the lemma under our exponent convention. All choices of $N$,
$K$, and the finite identity precede the limit $n\to\infty$.
\end{proof}

It remains to bound these weighted sums using the angular constraint.

\begin{proof}[Proof of Theorem~\ref{cor:general-exponents}]
Reorder the character's dot-product parameters so that the objective is
$ap_1+bp_2+cp_3$, and write $p_i=\cos^2\theta_i$ with
$\theta_i\in[0,\pi/2]$. By Proposition~\ref{prop:angular-cap}, their
sum is at least $\pi/2$. Decreasing angles to make the sum exactly
$\pi/2$ only increases the objective, so it suffices to bound that case.
If $b=0$, then $c=0$ and the objective is at most $a$.

For $a,b,c>0$, put $\beta_i=2\theta_i$ and
\[
 u=\sqrt{\frac{bc}{a}},\qquad
 v=\sqrt{\frac{ac}{b}},\qquad
 w=\sqrt{\frac{ab}{c}}.
\]
Thus $vw=a$, $wu=b$, $uv=c$, and $\beta_1+\beta_2+\beta_3=\pi$.
Wolstenholme's inequality
\cite[Equation~(1)]{liu-2020-weighted-trigonometric}, applied with weights
$u,v,w$, gives
\[
 2(a\cos\beta_1+b\cos\beta_2+c\cos\beta_3)\le u^2+v^2+w^2.
\]
The inequality also holds at zero angles by continuity.
Since $p_i=(1+\cos\beta_i)/2$, it gives
\[
 ap_1+bp_2+cp_3\le\frac{a+b+c}{2}
 +\frac{1}{4}\left(\frac{ab}{c}+\frac{ac}{b}+\frac{bc}{a}\right).
\]
For $a,b>0$, set $c_0=\frac{ab}{a+b}$. If $0\le c\le c_0$, replace $c$
in the objective by $c_0$, which only increases it. The displayed bound
at $c_0$ equals $a+b$. This proves both branches, including $c=0$.

We have bounded the objective for every tensor character.
Lemma~\ref{lem:cost-conversion}, with target
$\mathbf h=(c,a,b)$, gives \eqref{eq:general-exponents}.
\end{proof}

\begin{proof}[Proof of Corollary~\ref{thm:curve}]
Apply Theorem~\ref{cor:general-exponents} to the sorted dimensions
of $(1,k,1)$. For $0\le k\le1$, the first branch applies exactly when
$k\le\frac{1}{2}$. In the remaining range, both possible orderings give
\[
 \frac{2+k}{2}+\frac{1}{4}\left(2k+\frac1k\right)
 =1+k+\frac1{4k}.
\]
The output-size lower bound $\omega(k)\ge2$ proves equality for
$0\le k\le\frac{1}{2}$ and hence $\alpha\ge\frac{1}{2}$.
\end{proof}

Some values of the bound in Corollary~\ref{thm:curve} are:
\begin{center}
\begin{tabular}{ccccccc}
\toprule
$k$ & $0.5$ & $0.6$ & $0.8$ & $1$ & $1.5$ & $2$\\
\midrule
$\omega(k)\le$ & $2$ & $2.0166\ldots$ & $2.1125$ & $2.25$ & $2.6666\ldots$ & $3.125$\\
\bottomrule
\end{tabular}
\end{center}
The exponent convention gives an $O(n^2)$-time
algorithm at $k=\frac{1}{2}$.

\section{Applications}\label{sec:applications}

Corollary~\ref{thm:curve} and Theorem~\ref{cor:general-exponents} give the following bounds through standard
rectangular-multiplication balances. All bounds use arithmetic
operations, with graph operations and symbol comparisons also charged.

\begin{enumerate}
\item \textbf{One lowest common ancestor per pair in a DAG.}
Grandoni et al.\ \cite[Theorem~4.2 of the published version]
{grandoni-et-al-2021-all-pairs-lca} obtain exponent $1+2x$ from the
balance $\omega(x)=3x$. Corollary~\ref{thm:curve} permits
$x=(1+\sqrt3)/4$, giving
$O(n^{(3+\sqrt3)/2})$ time for finding one lowest common
ancestor for every pair in an $n$-vertex DAG. The exponent is
$(3+\sqrt3)/2\approx2.36603$.

\item \textbf{Exact all-pairs Hamming distance.}
For $n$ vectors of length $n$ over an unrestricted alphabet, the
reduction to dominance products~\cite{labib-uznanski-wolleb-graf-2019}
and Yuster's tradeoff
\cite[Proposition~2.1]{yuster-2009-permutations-dominance} give the balance
$\rho=\omega(4-\rho)$.
Corollary~\ref{thm:curve} gives
$\rho=(13-\sqrt7)/4\approx2.58856$ and running time
$O(n^\rho)$ for computing every exact Hamming distance.

\item \textbf{The standard unweighted APSP balance.}
Zwick's deterministic algorithm for directed unweighted APSP
\cite[Section~6]{zwick-2002-bridging-sets} has exponent $2+\mu$,
where $\omega(\mu)\le1+2\mu$. Here $\omega(\frac{1}{2})=2=1+2(\frac{1}{2})$,
so it gives $O(n^{2.5})$ running time.
Combining the full curve with the deterministic $(\min,+)$-product
algorithm and APSP reduction of Alman and Vassilevska Williams
\cite[Theorems~22 and~33--34]{alman-williams-2026} further gives
$O(n^{2.4999})$ running time.
\item \textbf{Minimum and maximum witnesses.}
For two Boolean $n\times n$ matrices, a minimum (maximum) witness for
entry $(i,j)$ is the least (greatest) index $k$ with $A_{ik}=B_{kj}=1$.
The standard blocking algorithm
\cite[Section~2.1]{cohen-yuster-2014-minimum-witnesses} takes
$O(n^{1-r+\omega(r)}+n^{2+r})$ time.
Taking $r=\frac{1}{2}$ gives $O(n^{2.5})$ time for all witnesses.

\item \textbf{Sparse matrix multiplication.}
For two $n\times n$ matrices with at most $m$ nonzero entries each,
the sparse--dense decomposition of Yuster and Zwick
\cite[Section~3]{yuster-zwick-2005-sparse} gives
$O(m^2/n^r+n^{\omega(r)})$ time for $0\le r\le1$.
Thus $m\le n^{1.25}$ permits $O(n^2)$ time.

\item \textbf{Dynamic matrix inverse and reachability.}
Substituting Theorem~\ref{cor:general-exponents} into van den
Brand--Nanongkai--Saranurak's tradeoff
\cite[Theorem~4.2]{van-den-brand-nanongkai-saranurak-2019}, with
$u=0.50323551$ and $v=0.80667763$ in
$\max\{u+v,\omega(1,u,v)-u,\omega(v)-v\}$,
gives $O(n^{1.30992})$ worst-case time per entry update and query for
maintaining the inverse of a nonsingular $n\times n$ matrix, after
$O(n^{2.25})$ preprocessing.
The associated Monte Carlo reduction gives the same bounds for directed
reachability under edge updates, with $O(1)$ query time for a fixed pair
\cite[Corollary~C.12 and Theorem~C.16]{van-den-brand-nanongkai-saranurak-2019}.
\end{enumerate}


\clearpage
\appendix
\makeatletter
\renewcommand{\@seccntformat}[1]{\appendixname~\csname the#1\endcsname.\quad}
\makeatother
\numberwithin{equation}{section}
\section{Splitting two polynomial products}\label{app:determinant-filtration}

We recap the two-sector construction used in
\cite[Lemma~4.1]{openai-2026-matrix-nine-fourths} as an explicit change of
coordinates for pairs of polynomials.\footnote{This is the degree-one case
of the classical Clebsch--Gordan formula
\cite[Theorem~2.6.3, pp.~56--57 of the author version]{kowalski-2014-representation-theory}.
The coefficient calculation here does not require that formula.}

Fix $r\ge1$ and $s\ge2$. The tensor $C_{XY}(r,s)\otimes B_X(2)$
computes
\[
 (f,(a,b))\longmapsto(fa,fb),\qquad
 \deg f<r,\quad\deg a,\deg b<s.
\]
Thus tensoring with $B_X(2)$ means multiplying two polynomials by the
same $f$. We will change coordinates to split the second input into
lengths $s+1$ and $s-1$.

For $e\ge1$ and $\deg a,\deg b\le e$, use
\[
 p=a+tb,\qquad h=b-p_{e+1}t^e,\qquad p_d=[t^d]p.
\]
Here $[t^d]p$ denotes the coefficient of $t^d$. We have
$\deg p\le e+1$ and $\deg h\le e-1$, and these coordinates are invertible:
$b=h+p_{e+1}t^e$, $a=p-tb$.
Use them on the input pair with $e=s-1$ and on the output pair with
$e=r+s-2$. Multiplication by $f$ becomes
\[
 (p,h)\longmapsto
 \bigl(fp,\,fh+p_s t^{s-1}(f-f_{r-1}t^{r-1})\bigr).
\]
Scale input $h$ by $\varepsilon^{-1}$ and its output coordinate by
$\varepsilon$. The map is now
\[
 (p,h)\longmapsto
 \bigl(fp,\,fh+\varepsilon p_s t^{s-1}(f-f_{r-1}t^{r-1})\bigr)
 \xrightarrow{\varepsilon\to0}(fp,fh).
\]
The independent polynomials $p,h$ have lengths $s+1,s-1$, respectively.
This gives $C_{XY}(r,s+1)$ and $C_{XY}(r,s-1)$ sharing the first leg,
with disjoint second and third legs. Permuting the legs gives the
other orientations.

\section{The three-sector splitting}\label{app:tripling}

We recap the construction from \cite[Lemma~4.2]{openai-2026-matrix-nine-fourths}.
In $C_{XY}(r,3s+r-1)$, a term has $0\le x<r$ and output index $z=x+y$.
Partition the second-input and output indices into three intervals:
\[
\begin{array}{c|ccc}
 &\text{left}&\text{middle}&\text{right}\\ \hline
 Y &[0,s-1]&[s,2s+r-2]&[2s+r-1,3s+r-2]\\
 Z &[0,s+r-2]&[s+r-1,2s+r-2]&[2s+r-1,3s+2r-3].
\end{array}
\]
Scale middle $Y$ variables by $\varepsilon$ and middle $Z$ variables
by $\varepsilon^{-1}$, leaving the others unchanged.
Since $z=y+x$ with $0\le x<r$, left and right $Y$ indices stay
in the corresponding $Z$ intervals. Cross-sector terms thus connect
middle $Y$ to outer $Z$ and acquire a factor $\varepsilon$;
matched terms remain unchanged. As $\varepsilon\to0$, precisely
three matched sectors sharing $X$ remain.

Translations identify the outer sectors with $C_{XY}(r,s)$. The middle is
\[
 \sum_{u=0}^{r-1}\sum_{v=0}^{s-1}
 x_{r-1-u}y_{s+u+v}z_{s+r-1+v},
\]
which is $C_{XZ}(r,s)$ after reversing its first-leg basis. This
identifies the individual sector. In the fixed-word construction of
\eqref{eq:type-shared}, the sectors acquire independent first-leg
variables before their individual bases are changed.

\section{Finite separation of a shared leg}\label{app:finite-separation}

We adapt the construction of \cite[Proposition~3.1]{openai-2026-matrix-nine-fourths}
with an equal-sums filter, reducing the copy count from $5M$ to $2M$.
It suffices to share leg $X$. Write
\[
 T=\sum_{h=1}^M\sum_{a,b,c}\mu_{habc}x_a y_{hb}z_{hc}.
\]
Take $L=2M$ copies and a primitive $L$-th root of unity $\zeta$.
In copy $\ell=0,\ldots,L-1$, make the substitutions
\[
\begin{aligned}
 x_a^{(\ell)}&\longmapsto\sum_{g=1}^M\zeta^{\ell g}X_{ag},\\
 y_{hb}^{(\ell)}&\longmapsto\sum_{u=1}^M\zeta^{\ell(u-h)}Y_{hbu},\\
 z_{hc}^{(\ell)}&\longmapsto\frac1L\sum_{v=1}^M
                        \zeta^{-\ell v}Z_{hcv}.
\end{aligned}
\]
Summing the copies multiplies each coefficient by
\[
 \frac1L\sum_{\ell=0}^{L-1}\zeta^{\ell(g+u-h-v)}.
\]
The geometric sum is one for exponent divisible by $L$ and zero
otherwise. Since $|g+u-h-v|\le2(M-1)<L$, it retains precisely
$g+u=h+v$.

Now give $X_{ag},Y_{hbu},Z_{hcv}$ respective weights
$g^2$, $2hu-h^2$, and $-2hv$. On every retained term their total is
\[
 g^2+2h(u-v)-h^2=(g-h)^2\ge0.
\]
The weight-zero terms have $g=h$ and $u=v$, so the limit is
\[
 \sum_{h,a,b,c,u}\mu_{habc}X_{ah}Y_{hbu}Z_{hcu}
 =\bigoplus_h(T_h\otimes B_X(M)).
\]
Every leg now distinguishes $h$. Although individual weights may be
negative, the transformed coefficients form a polynomial of degree at
most $(M-1)^2$. This proves \eqref{eq:finite-separation}.

\section{Interpolation of a degeneration}\label{app:degeneration-interpolation}

For completeness, we give the standard interpolation proof
\cite[Theorem~8]{christandl-jensen-zuiddam-2018-nonmultiplicativity}.

\begin{proof}[Proof of Lemma~\ref{lem:degeneration-interpolation}]
Let $P(\varepsilon)$ be the tensor obtained from $T$ by the degeneration's
substitutions. Its coefficients are polynomials in $\varepsilon$ of degree
at most $L$, and $P(0)=S$. Thus $L$ describes the chosen degeneration.
For every $\varepsilon\ne0$, $P(\varepsilon)$ is a restriction of $T$.
Put $Q(\varepsilon)=P(\varepsilon)^{\otimes N}$. Each coefficient of $Q$
is a product of $N$ coefficients of $P$, so its degree is at most $NL$,
and $Q(0)=S^{\otimes N}$.

Choose $NL+1$ distinct nonzero points $\varepsilon_0,\ldots,\varepsilon_{NL}$.
Lagrange interpolation, applied to every coefficient, gives
\[
 S^{\otimes N}=Q(0)=\sum_{j=0}^{NL}c_jQ(\varepsilon_j),\qquad
 c_j=\prod_{\substack{0\le h\le NL\\h\ne j}}
 \frac{-\varepsilon_h}{\varepsilon_j-\varepsilon_h}.
\]
Each $Q(\varepsilon_j)$ is a restriction of $T^{\otimes N}$.
To form the displayed weighted sum from their direct sum, identify the
corresponding variables on all three legs and scale the variables on
one leg of copy $j$ by $c_j$. These are allowed linear substitutions,
proving \eqref{eq:degeneration-interpolation}. Applying a character and
taking $N$-th roots gives
\[
 \lambda(S)\le(NL+1)^{1/N}\lambda(T).
\]
Letting $N\to\infty$ proves $\lambda(S)\le\lambda(T)$.
No continuity assumption on $\lambda$ is needed.
\end{proof}

\section{Replicating a convolution}\label{app:replication}

\begin{proof}[Proof of Lemma~\ref{lem:replication}]
Put $P(t)=t^m-1$. Repeated polynomial division writes the inputs uniquely as
\[
 F=\sum_{a=0}^{r-1}F_aP^a,\qquad
 G=\sum_{b=0}^{s-1}G_bP^b,\qquad \deg F_a,\deg G_b<m.
\]
These are invertible linear coordinates: the basis polynomials $t^dP^a$
have distinct leading degrees $ma+d$. On the output retain the analogous
blocks $H_c$ for $0\le c<r+s-1$. They use $m(r+s-1)$ independent
coordinates, within the original output length $m(r+s)-1$.
Divide each $F_aG_b=Q_{ab}+PR_{ab}$ with $\deg Q_{ab},\deg R_{ab}<m$.
Scale input blocks $a,b$ by $\varepsilon^{-a},\varepsilon^{-b}$ and output
block $c$ by $\varepsilon^c$. The transformed product has
\[
 H_c=\sum_{a+b=c}Q_{ab}
       +\varepsilon\sum_{a+b=c-1}R_{ab}.
\]
Let $\alpha_1,\ldots,\alpha_m$ be the distinct roots of $P$.
Since $Q_{ab}(\alpha_h)=F_a(\alpha_h)G_b(\alpha_h)$, the limit gives
\[
 H_c(\alpha_h)=\sum_{a+b=c}F_a(\alpha_h)G_b(\alpha_h)
 \qquad(1\le h\le m).
\]
Evaluation at $m$ distinct points is invertible on each block of degree
less than $m$. The resulting coordinates therefore give $m$ independent
convolutions of lengths $r,s$.
The evaluation step is the usual diagonalization of cyclic convolution
\cite[Theorem~5 and Section~5.1]{ju-solomonik-2020-convolution}; the displayed
scalings prove the degeneration, with coefficient degree at most one.
\end{proof}
\section{Counting and rescaling}\label{app:counting-rescaling}

We give the complete counting argument and justify the real-length
profiles used in Sections~\ref{sec:discrete}--\ref{sec:homogeneous}.

To prove \eqref{eq:dot-parameters}, fix an integer $q\ge1$ and put
$m=\lfloor N\log_2q\rfloor$. Monotonicity in dimension
and $B_i(q)^{\otimes N}\cong B_i(q^N)$ give
\[
 2^{mp_i}\le\lambda(B_i(q))^N\le2^{(m+1)p_i}.
\]
Take $N$-th roots and let $N\to\infty$. The interval $[0,1]$ follows
from $1\le\lambda(B_i(2))\le2$.

\begin{proof}[Proof of Lemma~\ref{lem:power-sum}]
For every tuple $\sum_h n_h=N$, applying $\lambda$ to
\eqref{eq:type-shared} gives
\[
 2M\lambda(T)^N\ge M^{1+p_i}\prod_h\lambda(T_h)^{n_h}.
\]
Put $a_h=\lambda(T_h)^{1/p_i}$. Cancel $M$ and take the power $1/p_i$:
\[
 M\prod_h a_h^{n_h}\le2^{1/p_i}\lambda(T)^{N/p_i}.
\]
There are at most $(N+1)^s$ multiplicity tuples. Summing and using the multinomial
theorem gives
\[
 \left(\sum_h a_h\right)^N
 \le(N+1)^s2^{1/p_i}\lambda(T)^{N/p_i}.
\]
For fixed $s$, this extra factor is polynomial in $N$. Taking
$N$-th roots and letting $N\to\infty$ proves \eqref{eq:power-sum}.
\end{proof}

\begin{proof}[Proof of Lemma~\ref{lem:discrete}]
The two-sector splitting and Lemma~\ref{lem:power-sum} give, for $s\ge2$,
\[
 2f_{ij}(r,s)^{1/p_i}
 =\lambda(C_{ij}(r,s)\otimes B_i(2))^{1/p_i}
 \ge f_{ij}(r,s+1)^{1/p_i}+f_{ij}(r,s-1)^{1/p_i}.
\]
This is discrete concavity in $s$; input interchange gives the other
coordinate. The three-sector splitting and the same lemma give
\eqref{eq:discrete-tripling}.
\end{proof}

\begin{proof}[Proof of Lemma~\ref{lem:homogeneous}]
Initially let $r,s>0$ be rational, and define
\begin{equation}\label{eq:rational-rescaled-profile}
 g_{ij}(r,s)=
 \sup_{\substack{q\in\mathbb Q_{>0}\\qr,qs\in\mathbb Z_{>0}}}
 \frac{f_{ij}(qr,qs)}{q}.
\end{equation}
Replication gives $f_{ij}(m,m)\ge m$, and monotonicity therefore gives
$f_{ij}(r,s)\ge\min\{r,s\}$ for integer lengths.
The bounds and input interchange follow from \eqref{eq:profile-rank}
and the lower bound above. Replacing $q$ by $q/t$ proves homogeneity.

The crucial point is that finitely many supremum values can be
approached at one common scale. Choose scales $q_h$ approximating
the desired values, and choose $Q$ so that every $Q/q_h$ is a positive
integer. By \eqref{eq:profile-replication},
\[
 \frac{f_{ij}(Qr,Qs)}{Q}
 \ge\frac{f_{ij}(q_hr,q_hs)}{q_h}.
\]
Replacing $Q$ by an integer multiple preserves all these approximations
and can clear any additional rational denominators. Monotonicity now follows by
comparing lengths at a common scale.

For concavity, take $s=\tau s_1+(1-\tau)s_2$, with rational
$0\le\tau\le1$. Approach $g_{ij}(r,s_1)$ and $g_{ij}(r,s_2)$ at
one scale $Q$, with $Qr,Qs_1,Qs_2,Qs$ integers. Discrete concavity gives
\[
 \left(\frac{f_{ij}(Qr,Qs)}{Q}\right)^{1/p_i}
 \ge\tau\left(\frac{f_{ij}(Qr,Qs_1)}{Q}\right)^{1/p_i}
 +(1-\tau)\left(\frac{f_{ij}(Qr,Qs_2)}{Q}\right)^{1/p_i}.
\]
Bound the left side by $g_{ij}(r,s)^{1/p_i}$ and let the approximation
error tend to zero. Input interchange gives the other coordinate.

For tripling, monotonicity first enlarges the discrete source length
$3s+r-1$ to $3s+r$. At a common scale $Q$,
\[
 \left(\frac{f_{ij}(Qr,Q(3s+r))}{Q}\right)^{1/p_i}
 \ge2\left(\frac{f_{ij}(Qr,Qs)}{Q}\right)^{1/p_i}
 +\left(\frac{f_{ik}(Qr,Qs)}{Q}\right)^{1/p_i}.
\]
Approach the two right-hand suprema simultaneously; the left side is
bounded by $g_{ij}(r,3s+r)^{1/p_i}$. This proves
\eqref{eq:general-tripling}.

To extend the profiles to real lengths, we first compare two positive
rational pairs $(r,s)$ and $(r',s')$. Put
\[
 a=\min\left\{\frac{r'}r,\frac{s'}s\right\},\qquad
 b=\max\left\{\frac{r'}r,\frac{s'}s\right\}.
\]
Homogeneity and monotonicity give
\[
 a\,g_{ij}(r,s)\le g_{ij}(r',s')\le b\,g_{ij}(r,s).
\]
Together with $g_{ij}(r,s)\le r+s$, this implies
\begin{equation}\label{eq:profile-continuity}
 |g_{ij}(r',s')-g_{ij}(r,s)|
 \le\frac{r+s}{\min\{r,s\}}
       \max\{|r'-r|,|s'-s|\}.
\end{equation}
The coefficient is bounded on every closed bounded region that stays
away from the coordinate axes. Approximating a positive real pair by
rationals therefore gives a unique continuous extension. All the rational
bounds, homogeneity, monotonicity, input interchange, and powered
concavity extend by continuity. So does \eqref{eq:general-tripling}.

We finally check that this extension agrees with the limit in
\eqref{eq:rescaled-profile}. For positive integers $A,B$, the scale $q=1$ in
\eqref{eq:rational-rescaled-profile} gives $f_{ij}(A,B)\le g_{ij}(A,B)$.
Thus for arbitrary real $r,s>0$,
\[
 \frac{f_{ij}(\lceil mr\rceil,\lceil ms\rceil)}m
 \le g_{ij}\left(\frac{\lceil mr\rceil}m,
                 \frac{\lceil ms\rceil}m\right)
 \longrightarrow g_{ij}(r,s),
\]
using homogeneity and continuity.

For the reverse bound, fix positive rationals $r_0<r$, $s_0<s$ and
an admissible scale $q$ with $qr_0,qs_0$ positive integers.
Put $k=\lfloor m/q\rfloor$. For all sufficiently large $m$, $k\ge1$,
and $kqr_0\le\lceil mr\rceil$, $kqs_0\le\lceil ms\rceil$.
Monotonicity and replication give
\[
 \frac{f_{ij}(\lceil mr\rceil,\lceil ms\rceil)}m
 \ge\frac{k}{m}f_{ij}(qr_0,qs_0).
\]
Since $k/m\to1/q$, taking the lower limit and then the supremum over
$q$ gives a lower bound $g_{ij}(r_0,s_0)$. Let $r_0,s_0$ approach
$r,s$ from below and use continuity. This proves the limit formula
\eqref{eq:rescaled-profile}, and completes the lemma.
\end{proof}

\section{Common slopes and the angular constraint}\label{app:slopes-angular}

Here we justify the common slopes and offsets, then derive the angular
constraint using finite increments rather than derivatives.

\begin{proof}[Proof of Lemma~\ref{lem:intercepts}]
Write $P=P_{ij}$ and $Q=P_{ik}$. We will define
\[
 c_{ij}=\sup_{s>0}\left(\frac{P(s)}{L_i}-s\right),\qquad
 c_{ik}=\sup_{s>0}\left(\frac{Q(s)}{L_i}-s\right)
\]
once the common positive slope $L_i$ has been established. For $s>1$, the nonnegative secant
slope $(P(s)-P(1))/(s-1)$ decreases as $s$ grows, by concavity.
It therefore has a finite limit, which is also $L_P=\lim P(s)/s$.
Likewise define
$L_Q=\lim Q(s)/s$. Dividing the coupled inequalities
\[
 P(3s+1)\ge2P(s)+Q(s),\qquad
 Q(3s+1)\ge2Q(s)+P(s)
\]
by $s$ and passing to infinity gives $L_P\ge L_Q$ and $L_Q\ge L_P$.
Call their common value $L_i$.

Their sum $R=P+Q$ satisfies $R(3s+1)\ge3R(s)$. Iterating at
$s_m=3^m(s+\frac{1}{2})-\frac{1}{2}$ gives
\begin{equation}\label{eq:sum-majorant}
 R(s)\le\lim_{m\to\infty}3^{-m}R(s_m)=L_i(2s+1).
\end{equation}
Since $R(s)>0$, we have $L_i>0$. For $v>u>0$, concavity bounds
the secant slope $(P(v)-P(u))/(v-u)$ below by every later secant slope;
sending its far endpoint to infinity gives
\[
 P(v)-P(u)\ge L_i(v-u).
\]
Letting $u\downarrow0$ and using $P(u)\ge0$ gives
$P(v)\ge L_i v$. Thus $P(s)/L_i-s$ is nonnegative and nondecreasing;
the same holds for $Q$. Their sum is at most one by
\eqref{eq:sum-majorant}. Approaching both suprema at one sufficiently
large $s$ proves $c_{ij}+c_{ik}\le1$. The affine upper bounds
follow from the definition of the offsets. Since the residuals are
nondecreasing and bounded, these suprema also equal their limits as
$s\to\infty$.
\end{proof}

\begin{proof}[Proof of Proposition~\ref{prop:angular-cap}]
If some $p_i=0$, its summand is already $\pi/2$. Assume all three
parameters are positive. Lemma~\ref{lem:intercepts}
gives six offsets and their three bounds
\[
 c_{ij}\ge0,\qquad c_{ij}+c_{ik}\le1.
\]
The increment and affine bounds give
\[
 \frac{P_{ij}(s+t)}{P_{ij}(s)}
 \ge1+\frac{L_it}{P_{ij}(s)}\ge1+\frac{t}{s+c_{ij}}.
\]
Using homogeneity, we obtain
\begin{equation}\label{eq:finite-increment}
 \frac{g_{ij}(r,s+t)}{g_{ij}(r,s)}
 \ge\left(1+\frac{t}{s+c_{ij}r}\right)^{p_i}.
\end{equation}
Increase $r$ by $\varepsilon r$, then $s$ by $\varepsilon s$.
Input interchange supplies the first bound and
\eqref{eq:finite-increment} supplies the second:
\[
\begin{aligned}
 1+\varepsilon
 &=\frac{g_{ij}((1+\varepsilon)r,(1+\varepsilon)s)}{g_{ij}(r,s)}\\
 &\ge\left(1+\frac{\varepsilon r}{r+c_{ji}s}\right)^{p_j}
       \left(1+\frac{\varepsilon s}{s+c_{ij}(1+\varepsilon)r}\right)^{p_i}.
\end{aligned}
\]
Take logarithms, divide by $\varepsilon$, and let
$\varepsilon\downarrow0$.
The elementary limit $\log(1+\varepsilon x)/\varepsilon\to x$ gives
\[
 1\ge\frac{p_i s}{s+c_{ij}r}+\frac{p_j r}{r+c_{ji}s}
 \qquad(r,s>0).
\]
With $t=r/s$, this is equivalent to
\[
 (1-p_j)c_{ij}t^2+
 (1-p_i-p_j+c_{ij}c_{ji})t+(1-p_i)c_{ji}\ge0
 \qquad(t>0).
\]
A quadratic $At^2+Bt+C$ with $A,C\ge0$ is nonnegative on $t>0$
exactly when $B\ge-2\sqrt{AC}$. Indeed, divide by $t$ and use
$\inf_{t>0}(At+C/t)=2\sqrt{AC}$; if $A$ or $C$ is zero, the
infimum is zero, approached at an endpoint. Thus
writing $x=\sqrt{c_{ij}c_{ji}}$ gives
$\bigl(x+\sqrt{(1-p_i)(1-p_j)}\bigr)^2\ge p_ip_j$, or
\[
 \sqrt{c_{ij}c_{ji}}\ge
 \left(\sqrt{p_ip_j}-\sqrt{(1-p_i)(1-p_j)}\right)_+,
\]
where $(z)_+=\max\{z,0\}$.

For any nonnegative $x_1,x_2,x_3$, AM--GM and
\eqref{eq:intercept-budgets} give
\[
 \sum_{i<j}2\sqrt{c_{ij}c_{ji}}x_ix_j
 \le\sum_{i<j}(c_{ij}x_i^2+c_{ji}x_j^2)
 \le\sum_i x_i^2.
\]
Put $\theta_i=\arccos\sqrt{p_i}$ and suppose
$S=\theta_1+\theta_2+\theta_3<\pi/2$. Take $x_i=\cos\theta_i>0$.
By \eqref{eq:intercept-products},
$\sqrt{c_{ij}c_{ji}}\ge\cos(\theta_i+\theta_j)$.
For the other two indices $j,k$, the product-to-sum identities give
\[
\begin{aligned}
 \cos(\theta_i+\theta_j)x_j+\cos(\theta_i+\theta_k)x_k
 &=\cos\theta_i+
   \frac{\cos(\theta_i+2\theta_j)+\cos(\theta_i+2\theta_k)}{2}\\
 &=x_i+\cos S\cos(\theta_j-\theta_k)>x_i.
\end{aligned}
\]
Both factors of the extra term are positive because $S<\pi/2$ and
$|\theta_j-\theta_k|\le S$.
Multiply by $x_i$ and sum over $i$:
\[
 \sum_{i<j}2\sqrt{c_{ij}c_{ji}}x_ix_j
 \ge\sum_{i<j}2\cos(\theta_i+\theta_j)x_ix_j
 >\sum_i x_i^2,
\]
contradicting \eqref{eq:angular-quadratic-form}.
\end{proof}
\end{document}